%% file: adjunction.tex
\documentclass[a4paper,12pt]{article}
\usepackage[math]{kurier}
\usepackage[T1]{fontenc}
\usepackage{graphicx}
\usepackage{enumitem}
\usepackage[
  top=1.5cm,
  left=1.5cm,
  right=1.5cm,
  bottom=2cm
]{geometry}

\usepackage{mathtools}
\usepackage{amsmath}
\usepackage{amsfonts}
\usepackage{amsthm}
\usepackage{amssymb}
\usepackage{xspace}
\usepackage{mathpartir}
\usepackage{tikz}
\usetikzlibrary{calc}

\input{macros}

\author{Paul Brunet}
\title{A 2-adjunction between representations\\ and\\ preorder morphisms}

\newtheorem{lemma}{Lemma}
\begin{document}

\maketitle
\section{Introduction}
The recently introduced model of representations has been defined and motivated somewhat ex-nihilo. In this document, I will show that representations are related to a more ``classical'' model through a 2-adjunction. The target model is that of preorder morphisms, i.e. maps between sets equipped with reflexive and transitive relation that satisfy some natural preservation property.

The aim of this is two-fold: first, this provides in my opinion a further justification of representations, as an object in non-trivial yet tight connection to some natural constructs; and secondly it suggests some classical results\footnote{Which ones? Yet to be investigated... But they surely should exist!}  about order preserving maps could have interesting consequences for representations.

\section{Preliminaries}
\paragraph{Notations}{\itshape\small
A binary relation $x:A\relation B$ is a subset of $A\times B$. We denote the set-theoretic inclusion of relations by $\leqrel$, writing $x\leqrel y$ instead of $x\subseteq y$. Similarly, equality of relations is denoted by $\eqrel$.
Given sets $A,B,C$, and relations $x:A\relation B$ and $y:B\relation C$, we define the following constants and operations: $1_A:A\relation A$ is the identity (diagonal) relation over $A$, $x\converse:B\relation A$ is the converse of $x$, and $x\then y:A\relation C$ is the sequential composition of $x$ and $y$, containing all pairs $(a,c)$ such that there is some $b\in B$ such that $(a,b)\in x$ and $(b,c)\in y$.
Given $z:A\to C$, the (left-)residual of $z$ by $x$, written $x\under z:B\relation C$, is the set of pairs $(b,c)$ such that for every $a\in A$, $(a,b)\in x$ entails $(a,c)\in z$. This operator has can also be defined by the following universal equivalence:
\begin{equation}
y\leqrel x\under z~\Leftrightarrow~ x\then y\leqrel z.\label{eq:galois}                                                                 \end{equation}
A dual operator $z\over y$, satisfying $x\leqrel z\over y$ iff $x\then y\leqrel z$ can be defined by $z\over y\eqdef\paren((y\converse)\under(z\converse))\converse$.
Given a function $f:A\to B$, the relations $f_\ast:A\relation B$ and $f^\ast:B\relation A$ are defined as the graph of the function (i.e. pairs $(a,f(a))$) and its converse (pairs $(f(a),a)$).
An important property we will leverage is the following (universal tautology):
$$f_\ast\then (x\under y)\then g^\ast\eqrel (x\then f^\ast)\under (y\then g^\ast).$$

To define explicit functions, we borrow the $\lambda$ notation from the eponymous calculus. In practice, given an expression $e$ with a free variable $x$, $\lambda x.e$ is the function that maps any $x_0$ to $e\brack[x\mapsto x_0]$, the expression where $x$ has been substituted with $x_0$.

We will make use several times of the fact that whenever we have a pair of functions $f,g:A\to 2^B$ whose co-domain is a powerset, we get $f=g$ iff ${\in}\then f^\ast\eqrel{\in}\then g^\ast$. This is a consequence of the fact that any such function can be written $$ f = \lambda a.\setcompr{b\in B}{b\in f(a)}=\setcompr{b\in B}{(b,a)\in \paren({\in\then f^\ast})}.$$

}
%\upshape

\subsection{The category of preorder morphisms}
  A \emph{preorder} is a reflexive and transitive relation $x:A\relation A$, over some set~$A$.
  In relation algebra, this is captured by the following axioms:
  \begin{align}
  1_A&\leqrel x\tag{reflexivity}\\
  x\then x&\leqrel x\tag{transitivity}
  \end{align}
  Alternatively, the single axiom $x \eqrel x\under x$ defines the same notion.
  Since $(x\under x)\under(x\under x)\eqrel x\under x$ (and this for any $x:A\relation B$), this means preorders are exactly relations of the shape $x\under x$.
  Notice in particular that ${\in}\under{\in}\eqrel{\subseteq}$.

  A \emph{(preorder)-morphism} (\emph{prom} for short) is a structure $\pom=\tuple<A,B,x,y,f>$ where $x:A\relation A$ and $y:B\relation B$ are preorders, and $f:A\to B$ is an order-preserving function, i.e.
  \begin{equation}
   f^\ast\then x\leqrel y\then f^\ast.\tag{order preservation}
  \end{equation}

  A \emph{prom-morphism} between a pair of proms $\pom=\tuple<A,B,x,y,f>$ and $\pom'=\tuple<A',B',x',y',f'>$ is a pair of functions $\phi,\psi$ such that $\tuple<A,A',x,x',\phi>$ and $\tuple<B,B',y,y',\psi>$ are proms, and furthermore the following identity holds (or equivalently the accompanying diagram commutes):
  \begin{center}
   \begin{tikzpicture}
   \node(id) at (-8,1){$\psi^\ast\then f^\ast\eqrel f'^\ast\then\phi^\ast$};
    \node(a)at (0,2) {$A$};
    \node(b)at (0,0) {$B$};
    \node(a')at (2,2) {$A'$};
    \node(b')at (2,0) {$B'$};
    \draw[->] (a) to node[midway,left]{$f$}(b);
    \draw[->] (a') to node[midway,right]{$f'$}(b');
    \draw[->] (a) to node[midway,above]{$\phi$}(a');
    \draw[->] (b) to node[midway,above]{$\psi$}(b');
    \end{tikzpicture}
  \end{center}

  The category $\PoM$ is thus defined as having proms as objects, and prom-morphisms as homomorphisms.

\subsection{Representations}
A representation is a structure $R=\tuple<M,S,\models,\leq>$ where $\leq:S\relation S$ is a preorder, and $\models:M\relation S$ is a relation such that:
\begin{equation}
    {\models}\then{\leq}\leqrel\models
    \tag{soundness}
\end{equation}
A representation morphism between two representations $R=\tuple<M,S,\models,\leq>$ and $R'=\tuple<M',S',\models',\leq'>$ is a pair $\phi,\tau$ such that $\tuple<S,S',\leq,\leq',\phi>$ is a prom and $\tau:M'\relation M$ is a relation such that:
\begin{center}
  \begin{tikzpicture}
  \node(id) at (-8,1){$\tau\then {\models}\eqrel {\models'}\then\phi^\ast$};
  \node(a)at (0,2) {$S$};
  \node(b)at (0,0) {$M$};
  \node(a')at (2,2) {$S'$};
  \node(b')at (2,0) {$M'$};
%     \draw[thick,->] (a) to node[midway,left]{$\models$}(b);
%     \draw[thick,->] (a') to node[midway,right]{$\models$}(b');
%     \draw[thick,->] (a) to node[midway,above]{$\phi$}(a');
%     \draw[thick,->] (b) to node[midway,above]{$\tau$}(b');
  \relarrow(b)[midway,left][$\models$](a);
  \relarrow(b')[midway,right][$\models'$](a');
  \relarrow(a')[midway,above][$\phi^\ast$](a);
  \relarrow(b')[midway,above][$\tau$](b);
  \end{tikzpicture}
\end{center}

Given two representations $R$ and $R'$, the set of representation morphisms between them can be ordered by setting $\phi,\tau\leqslant\phi',\tau'$ iff $\phi=\phi'$ and $\tau\leqrel\tau'$.

Therefore we set $\Repr$ to be the 2-category, or in this case the order-enriched category, with representations as 0-cells (objects), representation morphisms as 1-cells, and $\leqslant$ as the collection of 2-cells (the order structure on the hom-sets).

\section{The functors}
\subsection{From proms to representations}
\[\R:\PoM\to\Repr\]

Let $\pom=\tuple<A,B,x,y,f>$ be a prom, we define $\R\pom$ to be the representation
\[\R\pom\eqdef\tuple<B,A,y\then f^\ast,x>.\]

Given a pair of proms $\pom,\pom'$ and a prom morphism $\phi,\psi$ between them, we define the following representation morphism between $\R\pom$ and $\R\pom'$:
\[\R(\phi,\psi)\eqdef\phi,\paren(y'\then\psi^\ast)\]

$\R$ is a 2-functor from $\PoM$ to $\Repr$, which we prove in the next three lemmas.

\begin{lemma}
 $\R\pom$ is a representation.
\end{lemma}
\begin{proof}
  \[\models\then\leq\eqrel y\then f^\ast\then x\leqrel y\then y\then f^\ast\leqrel y\then f^\ast\eqrel \models.\]
\end{proof}
\begin{lemma}
 $\R(\phi,\psi)$ is a representation morphism.
\end{lemma}
\begin{proof}
  \[
    {\models'}\then\phi^\ast\eqrel y'\then f'^\ast\then \phi^\ast\eqrel y'\then\psi^\ast\then f^\ast\eqrel\tau\then f^\ast.
  \]
  We show that $\tau\eqrel\tau\then y$:
  \begin{itemize}
   \item by reflexivity of $y$ we have $\tau \leqrel \tau\then y$
   \item $\tau\then y\eqrel y'\then\psi^\ast\then y\leqrel y'\then y'\then\psi^\ast\leqrel y'\then\psi^\ast\eqrel\tau$.
  \end{itemize}
  Therefore ${\models'}\then\phi^\ast\eqrel\tau\then f^\ast\eqrel\tau\then y\then f^\ast\eqrel\tau\then{\models}$.
\end{proof}

\begin{lemma}
 $id_A,1_B\leqslant\R(id_A,id_B)$, and furthermore, given $\phi,\psi:\pom\to\pom'$ and $\phi',\psi':\pom'\to\pom''$, we have
 \[\R(\paren(\phi',\psi')\circ\paren(\phi,\psi))\leqslant\R(\phi',\psi')\circ\R(\phi,\psi).\]
\end{lemma}
\begin{proof}
  \begin{align*}
   \R(id_A,id_B)
   &={id_A,y}\geqslant{id_A,1_B}\tag{reflexivity of $y$}\\
   \R({\phi',\psi'}\circ{\phi,\psi})
   &=\paren(\phi'\circ\phi,y''\then\psi'^\ast\then\psi^\ast)\\
   &\leqslant\paren(\phi'\circ\phi,y''\then\psi'^\ast\then y'\then\psi^\ast)\\
   &=\paren(\phi',y''\then\psi'^\ast)\circ\paren(\phi',y'\then\psi^\ast)\\
   &=\R(\phi',\psi')\circ\R(\phi,\psi).\tag*\qedhere
  \end{align*}
\end{proof}
\subsection{From representations to proms}
\[\M:\Repr\to\PoM\]
Let $R=\tuple<M,S,\models,\leq>$ be a representation. We define
\[\M*(R)=\tuple<S,2^M,\leq,\subseteq,\lambda s.\setcompr{m\in M}{m\models s}>. \]
Observe that ${\in} \then f_{\M*(R)}^\ast\eqrel\models$.

Now, let $R,R'$ be a pair of representations, and $\phi,\tau$ a morphism between them. We define:
\[\M(\phi,\tau)=\phi,\lambda \alpha.\setcompr{b\in B'}{\exists a\in\alpha:\;(b,a)\in\tau}.\]
In the following we write $\M*(\tau)$ for $\lambda \alpha.\setcompr{b\in B'}{\exists a\in\alpha:\;(b,a)\in\tau}$. Observe that ${\in}\then\M*(\tau)^\ast\eqrel\tau\then{\in}$.

$\M$ is a functor from the category $\Repr$ (forgetting about the 2-cell structure) to the category $\PoM$.
We prove this in the next three lemmas.
\begin{lemma}
 $\M*(R)$ is a prom.
\end{lemma}
\begin{proof}
  First, observe that $y_{\M*(R)}\then f_{\M*(R)}^\ast\eqrel {\subseteq}\then f_{\M*(R)}^\ast\eqrel{\in}\under{\in}\then f_{\M*(R)}^\ast\eqrel{\in}\under{\paren({\in}\then f_{\M*(R)}^\ast)}\eqrel{\in}\under\models$.
  Therefore we have that $ f_{\M*(R)}^\ast\then x_{\M*(R)}\leqrel y_{\M*(R)}\then f_{\M*(R)}^\ast$ iff
  ${\in}\then f_{\M*(R)}^\ast\then x_{\M*(R)}\leqrel\models$.
  This is obvious from the definitions:
  \[{\in}\then f_{\M*(R)}^\ast\then x_{\M*(R)}\eqrel{\models}\then{\leq}\leqrel\models\]
\end{proof}
\begin{lemma}
 $\M(\phi,\tau)$ is a prom morphism.
\end{lemma}
\begin{proof}
First we check that $\tuple<2^M,2^{M'},\subseteq,\subseteq,\M*(\tau)>$ is a prom, i.e. $\M*(\tau)^\ast\then{\subseteq}\leqrel{\subseteq}\then\M*(\tau)^\ast$:
\begin{align*}
\text{Since }&&{\subseteq}\then\M*(\tau)^\ast&\eqrel\paren({\in}\under{\in})\then\M*(\tau)^\ast\eqrel{\in}\under\paren({\in}\then\M*(\tau)^\ast),\\
\text{we get:}&&
  \M*(\tau)^\ast\then{\subseteq}\leqrel{\subseteq}\then\M*(\tau)^\ast
  &\Leftrightarrow {\in}\then \M*(\tau)^\ast\then{\subseteq}\leqrel{\in}\then\M*(\tau)^\ast\\
\text{We conclude }&&
  {\in}\then \M*(\tau)^\ast\then{\subseteq}
  &\eqrel\tau\then{\in}\then{\subseteq}\leqrel\tau\then{\in}\eqrel{\in}\then\M*(\tau)^\ast.
\end{align*}

Now we check $\psi^\ast\then f^\ast\eqrel f'^\ast\then\phi^\ast$, i.e. in our case:
\[\M*(\tau)^\ast\then f_{\M R}^\ast\eqrel f_{\M R'}^\ast\then\phi^\ast.\]
Because its domain is $2^{M'}$, this identity is equivalent to the following:
\[{\in}\then \M*(\tau)^\ast\then f_{\M R}^\ast\eqrel {\in}\then f_{\M R'}^\ast\then\phi^\ast.\]
We establish it:
\begin{align*}
 {\in}\then \M*(\tau)^\ast\then f_{\M R}^\ast
 &\eqrel \tau\then{\in}\then  f_{\M R}^\ast\\
 &\eqrel \tau\then{\models}\\
 &\eqrel {\models'}\then\phi^\ast\\
 &\eqrel {\in}\then f_{\M R'}^\ast\then\phi^\ast.\tag*\qedhere
\end{align*}
\end{proof}
\begin{lemma}
 $\M(id_A,1_M)=id_A,id_{2^M}$ and for every pair of morphisms $\phi,\tau:R\to R'$ and $\phi',\tau':R'\to R''$:
 \[\M(\paren(\phi',\tau')\circ\paren(\phi,\tau))=\M(\phi',\tau')\circ\M(\phi,\tau).\]
\end{lemma}
\begin{proof}
The identity $\M(id_A,1_M)=id_A,id_{2^M}$ boils down to $\M*(1_M)=id_{2^M}$, i.e. ${\in}\then\M*(1_M)^\ast\eqrel{\in}\then id^\ast_{2^M}$:
 \begin{align*}
  {\in}\then\M*(1_M)^\ast&\eqrel 1_M\then{\in}\eqrel\in\eqrel{\in}\then id^\ast_{2^M}
 \end{align*}
 For the other one, $\M(\paren(\phi',\tau')\circ\paren(\phi,\tau))=\M(\phi',\tau')\circ\M(\phi,\tau)$ amounts to
 \[{\in}\then\M(\tau'\then\tau)^\ast\eqrel{\in}\then\M*(\tau')^\ast\then\M*(\tau)^\ast,\]
 which we now establish:
 \begin{align*}
  {\in}\then\M(\tau'\then\tau)^\ast
  &\eqrel\tau'\then\tau\then{\in}\\
  &\eqrel\tau'\then{\in}\then\M(\tau)^\ast\\
  &\eqrel{\in}\then\M*(\tau')^\ast\then\M*(\tau)^\ast.\tag*\qedhere
 \end{align*}
\end{proof}

\section{The adjunction}
\subsection{As a pair of natural transformations}

\subsubsection{The unit}
\[\eta:Id_\PoM\Rightarrow\M*(\R)\]
Let us first compute an explicit description of $\M\circ\R(\pom)$, with $\pom=\tuple<A,B,x,y,f>$:
\begin{align*}
\M\circ\R(\pom)=\M\tuple<B,A,y\then f^\ast,x>
&=\tuple<A,2^B,x,\subseteq,\lambda a.\setcompr{b\in B}{(b,a)\in y\then f^\ast}>\\
&=\tuple<A,2^B,x,\subseteq,\lambda a.\setcompr{b\in B}{(b,f(a))\in y}>
\end{align*}
In other symbols:
\[{\in}\then\phi^\ast_{\M\circ\R(\pom)}\eqrel y\then f^\ast\]

Now, we define $\eta:\pom\to\M\R(\pom)$ as the pair $id_A,\lambda b.\setcompr{b'\in B}{(b',b)\in y}$.

We check the two requirements for $\eta$ to be well-defined : 1) it should be a pom-morphism, and 2) it should be natural.

\begin{description}
 \item[$\eta$ is a prom-morphism]  three requirements:
 \begin{description}
  \item[$id_A$ is a prom] trivial.
  \item[$\lambda b.\setcompr{b'\in B}{(b',b)\in y}$ is a prom]
  \begin{align*}
  \paren(\lambda b.\setcompr{b'\in B}{(b',b)\in y})^\ast\then y
  &\leqrel {\subseteq} \then  \paren(\lambda b.\setcompr{b'\in B}{(b',b)\in y})^\ast\then y\\
  &\eqrel {\in}\under\paren({\in} \then  \paren(\lambda b.\setcompr{b'\in B}{(b',b)\in y})^\ast)\then y\\
  &\eqrel \paren({\in}\under y)\then y
  \leqrel {\in}\under y
  \tag{$\dagger$}
  \\
  &\eqrel {\in}\under\paren(\in\then\paren(\lambda b.\setcompr{b'\in B}{(b',b)\in y})^\ast)\\
  &\eqrel {\subseteq}\then\paren(\lambda b.\setcompr{b'\in B}{(b',b)\in y})^\ast.
\end{align*}
For the magical step $\dagger$, we observe that $\paren({\in}\under y)\then y
  \leqrel {\in}\under y$ iff:
  \begin{align*}
   {\in}\then \paren({\in}\under y)\then y
   &\leqrel y
  \end{align*}
  Which holds because ${\in}\then \paren({\in}\under y)\leqrel y$ and $y$ is transitive.
 \item[the diagram] because the co-domain of the diagram (its sink) is a powerset,
 $\psi^\ast\then f^\ast\eqrel f'^\ast\then\phi^\ast$ is equivalent to
  ${\in}\then\psi^\ast\then f^\ast\eqrel {\in}\then f'^\ast\then\phi^\ast$, meaning in our case:
  \begin{align*}
   {\in}\then\psi^\ast\then f^\ast
   &\eqrel{\in}\then\paren(\lambda b.\setcompr{b'\in B}{(b',b)\in y})^\ast\then f^\ast\\
   &\eqrel y\then f^\ast\eqrel {\in}\then f'^\ast\then 1_A\eqrel {\in}\then f'^\ast\then\phi^\ast.
  \end{align*}
 \end{description}
 \item[$\eta$ is a natural transformation] here we only need to check one diagram. Consider a prom-morphism $\phi,\psi$. We have to verify
 \[\eta\circ\paren(\phi,\psi) = \M*(\R(\phi,\psi))\circ\eta. \]
 To make sense of it, we compute $\M*(\R(\phi,\psi))$:
 \begin{align*}
 \M*(\R(\phi,\psi))
 =\M(\phi,y'\then\psi^\ast)
 &=\paren(\phi,\lambda \alpha.\setcompr{b'\in B'}{\exists b\in\alpha:\;(b',b)\in y'\then\psi^\ast})\\
 &=\paren(\phi,\lambda \alpha.\setcompr{b'\in B'}{\exists b\in\alpha:\;(b',\psi(b))\in y'}).
 \end{align*}
 Hence we have:
 \begin{align*}
  \eta\circ\paren(\phi,\psi)
  &=\paren(id_A,\lambda b.\setcompr{b'\in B'}{(b',b)\in y'})\circ\paren(\phi,\psi)\\
  &=\paren(\phi,\lambda b.\setcompr{b'\in B'}{(b',b)\in y'}\circ\psi)\\
  &=\paren(\phi,\lambda b.\setcompr{b'\in B'}{(b',\psi(b))\in y'})\\
  &=\paren(\phi,\lambda b.\setcompr{b'\in B'}{(b',b)\in y'\then\psi^\ast})\\
  &=\paren(\phi,\lambda b.\setcompr{b'\in B'}{(b',b)\in y'\then \psi^\ast\then y}))\tag{$\dagger$}\\
  &=\paren(\phi,\lambda b.\setcompr{b'\in B'}{\exists b''\in B:\;(b'',b)\in y \text{ and }(b',\psi(b''))\in y'}))\\
  &=\paren(\phi,\lambda b.\setcompr{b'\in B'}{\exists b''\in\setcompr{b''\in B}{(b'',b)\in y}:\;(b',\psi(b''))\in y'}))\\
  &=\paren(\phi,\lambda \alpha.\setcompr{b'\in B'}{\exists b''\in\alpha:\;(b',\psi(b''))\in y'})\circ\lambda b.\setcompr{b''\in B}{(b'',b)\in y})\\
  &=\paren(\phi,\lambda \alpha.\setcompr{b'\in B'}{\exists b''\in\alpha:\;(b',\psi(b''))\in y'})\circ\paren(id_A,\lambda b.\setcompr{b''\in B}{(b'',b)\in y})\\
  &=\M*(\R(\phi,\psi))\circ\eta.
 \end{align*}
For the magical step $\dagger$, we prove by double inclusion that $
 y'\then \psi^\ast\eqrel y'\then \psi^\ast\then y$
\begin{align*}
 y'\then \psi^\ast&\leqrel y'\then \psi^\ast\then y\tag{by reflexivity of $y$}
\end{align*}
For the other inclusion we use the fact that $\psi$ is a prom, and transitivity of $y'$:
\[y'\then \psi^\ast\then y\leqrel y'\then y'\then\psi^\ast\leqrel y'\then\psi^\ast.\]
\end{description}

\subsubsection{The co-unit}
\[\epsilon:\R*(\M)\Rightarrow Id_\Repr\]
Let $R=\tuple<M,S,\models,\leq>$ be a representation, consider $\R\circ\M*(R)$:
\begin{align*}
 \R\circ\M*(R) &= \R\tuple<S,2^M,\leq,\subseteq,\lambda s.\setcompr{m\in M}{m\models s}>\\
 &=\tuple<2^M,S,{\subseteq}\then{\paren(\lambda s.\setcompr{m\in M}{m\models s})^\ast},\leq>
\end{align*}
We observe the following:
\begin{align*}
 {\subseteq}\then{\paren(\lambda s.\setcompr{m\in M}{m\models s})^\ast}
&\eqrel \paren({\in}\under{\in})\then{\paren(\lambda s.\setcompr{m\in M}{m\models s})^\ast}\\
&\eqrel {\in}\under{\paren({\in}\then{\paren(\lambda s.\setcompr{m\in M}{m\models s})^\ast})}\\
 &\eqrel {\in}\under{\models}\\
\end{align*}
Therefore
\[\R\circ\M*(R) =\tuple<2^M,S,{\in}\under{\models},\leq>.\]

We now define $\epsilon:\R*(\M*(R))\to R$ as the pair $id_S,\in$.
As in the case of the unit, we need to check that $\epsilon$ is 1) a morphism and 2) a natural transformation.
\begin{description}
 \item [$\epsilon$ is a morphism] two requirements:
 \begin{description}
  \item [$id_S$ is a prom] trivial;
  \item [the diagram] $\in\then {\in}\under{\models}\eqrel{\models}\then id^\ast$.
  Note that since $\leq$ does not appear in this formula, it means there are no restrictions on $\models$. In other words the statement can be generalized as the following lemma:
  \begin{lemma}\label{lem:astuce}
   For any sets $A,B$ and $x:A\relation B$, $x\eqrel{\in}\then{\paren({\in}\under{x})}$.
  \end{lemma}
  \begin{proof}
  The $\geqrel$ direction is obvious.
  For the other one, we use $\eta_{\pset}=\lambda a.\set{a}$, the unit of the powerset monad. It has the property that ${\in}\then\eta_{\pset}^\ast\eqrel 1_A$.
So, we have $x\eqrel{\in}\then\eta_{\pset}^\ast\then x$, thus ${\in}\then\eta_{\pset}^\ast\then x\leqrel x$ hence $\eta_{\pset}^\ast\then x\leqrel {\in}\under x$.
  This entails $x\eqrel{\in}\then\eta_{\pset}^\ast\then x\leqrel{\in}\then {\in}\under x$.
\end{proof}
%
%   Therefore we derive:
%   \begin{align*}
%    {\models}
%    &\eqrel 1_M\then{\models}\eqrel {\in}\then\eta^\ast_{\pset}\then{\models}.
%   \end{align*}
%   In particular this means ${\in}\then\eta^\ast_{\pset}\then{\models}\leqrel{\models}$, hence $\eta_{\pset}^\ast\then{\models}\leqrel{\in}\under{\models}$.
%   We may thus conclude:
%   \begin{align*}
%    {\models}
%    &\eqrel {\in}\then\eta^\ast_{\pset}\then{\models}
%     \leqrel {\in}\then{\in}\under{\models}.
%   \end{align*}
 \end{description}
\item[$\epsilon$ is natural]
  Let $\phi,\tau$ be a morphism from $R$ to $R'$. We need to check that:
  \[\epsilon\circ\R*(\M(\phi,\tau))=\paren(\phi,\tau)\circ\epsilon.\]
  Recall that $\R*(\M(\phi,\tau))=(\phi,{\subseteq}\then\M*(\tau)^\ast)$, and observe that ${\subseteq}\then\M*(\tau)^\ast\eqrel{\in}\under{\paren(\in\then\M*(\tau)^\ast)}\eqrel{\in}\under{\paren(\tau\then{\in})}$.
  As such, we get (using Lemma~\ref{lem:astuce})
  \[\epsilon\circ\R*(\M(\phi,\tau))=(\phi,{\in}\then{\in}\under{\paren(\tau\then{\in})})
  =\paren(\phi,\tau\then\in)
  =\paren(\phi,\tau)\circ\epsilon.\]

%   \begin{align*}
%
%   \end{align*}
\end{description}

\subsubsection{The diagrams}

There are two diagrams in the definition of an adjunction, one in each participating category.
Because \Repr is a 2-category, the corresponding diagram should be a 2-cell, i.e. an inclusion rather than an equality.
\begin{enumerate}
 \item $\epsilon\circ \R*(\eta) \geqslant id_{\R*(\pom)}$

\begin{align*}
 \epsilon\circ \R*(\eta)
 &= \paren(id_{A},\in)\circ\R(id_A,\lambda b.\setcompr{b'}{b\mathrel{y}b'})\\
 &= \paren(id_{A},\in)\circ\paren(id_A,{\subseteq}\then\lambda b.\setcompr{b'}{b\mathrel{y}b'}^\ast)\\
 &= \paren(id_{A},\in)\circ\paren(id_A,{\in}\under\paren({\in}\then\lambda b.\setcompr{b'}{b\mathrel{y}b'}^\ast))\\
 &= \paren(id_{A},\in)\circ\paren(id_A,{\in}\under y)\\
 &= \paren(id_{A}\circ id_A,{\in}\then{\in}\under y)\\
 &= \paren(id_A,y)\geqslant \paren(id_A,1_B)=id_{\R*(\pom)}.
\end{align*}

 \item $\M*(\epsilon)\circ \eta = id_{\M*(R)}$
 \begin{align*}
  \M*(\epsilon)\circ \eta
  &= \M(id_{S},\in)\circ\paren(id_S,\lambda X.\setcompr{Y}{X\subseteq Y})\\
  &= \paren(id_{S},\M*(\in))\circ\paren(id_S,\lambda X.\setcompr{Y}{X\subseteq Y})\\
  &= \paren(id_{S},\M*(\in)\circ\lambda X.\setcompr{Y}{X\subseteq Y})
 \end{align*}
 To conclude we need to show that $\M*(\in)\circ\lambda X.\setcompr{Y}{X\subseteq Y}=id_{2^M}$.
 Because these are functions with $2^M$ as codomain, their equality is equivalent to $[{\in}\then\paren(\M*(\in)\circ\lambda X.\setcompr{Y}{X\subseteq Y})^\ast\eqrel{\in}\then id_{2^M}^\ast$, which we now prove:
 \begin{align*}
  {\in}\then\paren(\M*(\in)\circ\lambda X.\setcompr{Y}{X\subseteq Y})^\ast
  &\eqrel{\in}\then\M*(\in)^\ast\then\paren(\lambda X.\setcompr{Y}{X\subseteq Y})^\ast\\
  &\eqrel{\in}\then{\in}\then\paren(\lambda X.\setcompr{Y}{X\subseteq Y})^\ast\\
  &\eqrel{\in}\then{\subseteq}\eqrel\in\eqrel{\in}\then id_{2^M}^\ast
 \end{align*}
\end{enumerate}

\subsection{The Galois connection}
Adjunctions can also be seen as correspondences between hom-sets, as described below.
Let $\pom$ be prom and $R$ a representation. We want to investigate the following connection:
\[ \mprset{fraction={===}}
  \inferrule{\R*(\pom)\to R}{\pom\to\M*(R)}\]

  We shall now define transformations $\Psi:\paren(\R*(\pom)\to R)\to\paren(\pom\to\M*(R))$ and
  $T:\paren(\pom\to\M*(R))\to\paren(\R*(\pom)\to R)$ between morphisms. Since the first components of both kinds of morphisms (ones between proms and between representations) are of the same kind, we define the first components of $\Psi(\phi,\tau)$ and $T(\phi,\psi)$ to be the function $\phi$ itself. As such we shall instead define $\Psi:(M\relation B)\to(B\to 2^M)$ and $T:(B\to 2^M)\to(M\relation B)$, and then set:
  \begin{align*}
   \Psi(\phi,\tau)&\eqdef (\phi,\Psi\tau)
   &T(\phi,\psi)& \eqdef(\phi,T\psi)
  \end{align*}
  To that effect:
  \begin{align*}
   \Psi:(M\relation B)&\to(B\to 2^M)
   &T:(B\to 2^M)&\to(M\relation B)\\
   \tau&\mapsto\lambda b.\setcompr{m\in M}{\exists b':\,(m,b')\in\tau\wedge(b',b)\in y}
   &\psi&\mapsto{\in}\then\psi^\ast.
   \end{align*}
   Observe that ${\in}\then\Psi\tau^\ast\eqrel\tau\then y$.
   \begin{lemma}
    If $(\phi,\psi)$ and $(\phi,\tau)$ are morphisms of the appropriate kind, so are $\Psi(\phi,\tau)$ and $T(\phi,\psi)$.
   \end{lemma}
   \begin{proof}
    First, we check that $\Psi\tau$ (equipped with the appropriate orders) is a prom.
    \begin{align*}
    \text{Since }{\subseteq}\then\Psi\tau^\ast&\eqrel\paren({\in}\under{\in})\then\Psi\tau^\ast\eqrel{\in}\under\paren({\in}\then\Psi\tau^\ast)\eqrel{\in}\under\paren(\tau\then y)\\
      \Psi\tau^\ast\then y\leqrel\subseteq\then\Psi\tau^\ast
      &\Leftrightarrow
              {\in}\then\Psi\tau^\ast\then y\leqrel\tau\then y\\
    \text{which holds trivially since }
    {\in}\then\Psi\tau^\ast\then y&\eqrel\tau\then y\then y\eqrel \tau\then y.
    \end{align*}

    The fact that $(\phi,\psi)$ and $(\phi,\tau)$ are morphisms means that:
    \begin{align*}
     &{\models}\then\phi^\ast\eqrel\tau\then\models_{\R*(\pom)}\eqrel\tau\then y\then f^\ast
     &&{\in}\then\psi^\ast\then f^\ast\eqrel {\in}\then f_{\M*(R)}^\ast\then\phi^\ast\eqrel{\models}\then{\phi^\ast}.
    \end{align*}
    We observe the following:
    \begin{align*}
     {\in}\then\psi^\ast \then f^\ast&\leqrel{\in}\then\psi^\ast\then y \then f^\ast \\
     {\in}\then\psi^\ast \then y\then f^\ast
     &\leqrel{\in}\then{\subseteq}\then\psi^\ast\then f^\ast\\
     &\leqrel{\in}\then\psi^\ast \then f^\ast\\
  \text{i.e. }{\in}\then\psi^\ast \then f^\ast&\eqrel{\in}\then\psi^\ast\then y \then f^\ast.
    \end{align*}

    Finally we check the identities required for $\Psi(\phi,\tau)$ and $T(\phi,\psi)$ to be morphisms:
    \begin{align*}
     {\models}\then\phi^\ast
     &\eqrel {\in}\then\psi^\ast \then f^\ast\\
     &\eqrel {\in}\then\psi^\ast\then y \then f^\ast\\
     &\eqrel T\psi\then\models_{\R*(\pom)}\\
     {\in}\then\Psi\tau^\ast\then f^\ast&\eqrel \tau\then y\then f^\ast\\
     &\eqrel {\models}\then\phi^\ast\\
     &\eqrel {\in}\then f_{\M*(R)}\then\phi^\ast\tag*\qedhere
    \end{align*}
   \end{proof}
   \begin{lemma}
    $\Psi T(\phi,\psi) = (\phi,\psi)$ and $(\phi,\tau)\leqslant T\Psi(\phi,\tau)$.
   \end{lemma}
   \begin{proof}
    \begin{align*}
     {\in}\then\Psi T\psi^\ast
     &\eqrel T\psi\then y\eqrel {\in}\then\psi^\ast\then y
     \end{align*}
    Observe that
    \begin{align*}
     {\in}\then\psi^\ast\then y &\geqrel {\in}\then \psi^\ast\tag{reflexivity of $y$}\\
      {\in}\then\psi^\ast\then y
      &\leqrel{\in}\then{\subseteq}\then\psi^\ast\tag{$\psi$ is a prom}\\
      &\eqrel{\in}\then \psi^\ast.
     \end{align*}
    Therefore
    \begin{align*}
     {\in}\then\Psi T\psi^\ast
     &\eqrel{\in}\then\psi^\ast
     \end{align*}
    For the diagram in \Repr:
    \begin{align*}
      T\Psi\tau
      &\eqrel{\in}\then\Psi\tau^\ast\eqrel\tau\then y\geqrel \tau.\tag*\qedhere
    \end{align*}

   \end{proof}
   Again we observe that the diagram in \Repr does not hold as an identity, but as an inequation, i.e. a 2-cell.

   \section{What about exactness ?}
   A key notion in the theory of representations is that of \emph{exact representations}. These are intuitiveley representations   where the preorder captures exactly the satisfaction relation, in that the following should hold:
   \[{\models}\under{\models}\leqrel{\leq}\]
   Since soundness (i.e. ${\models}\then{\leq}\leqrel{\models}$) is equivalent to the converse entailment,
   for exact representations this defining inequation actually holds as an identity.

   We now show that exactness corresponds to the well known notion of \emph{order-reflecting morphisms}: a prom $\pom=\tuple<A,B,x,y,f>$ is order-reflecting if $f_\ast\then y\then f^\ast\leqrel x$, i.e. if for every $a,a'\in A$, we have $(f(a),f(a'))\in y\Rightarrow (a,a')\in x$. Note that the fact that $f$ is order-preserving ($f^\ast\then x\leqrel y\then f^\ast$) entails the converse implication, meaning an order-reflecting prom actually satisfies $x\eqrel f_\ast\then y\then f^\ast$.

   \begin{lemma}
    $R$ is exact iff $\M*(R)$ is order-reflecting.
   \end{lemma}
   \begin{proof}
   Observe the following:
     \begin{align*}
      f_{\M*(R)\ast}\then y_{\M*(R)}\then f_{\M*(R)}^\ast
      &\eqrel
      \paren(\lambda s.\setcompr{m}{m\models s})_\ast
      \then{\subseteq}\then
      \paren(\lambda s.\setcompr{m}{m\models s})^\ast\\
      &\eqrel
      \paren({\in}\then\paren(\lambda s.\setcompr{m}{m\models s})^\ast)\under
      \paren({\in}\then\paren(\lambda s.\setcompr{m}{m\models s})^\ast)\\
      &\eqrel
      {\models}\under{\models}.
     \end{align*}
Therefore $R$ is exact means ${\models}\under{\models}\leqrel{\leq}$ which is thus equivalent to $f_{\M*(R)\ast}\then y_{\M*(R)}\then f_{\M*(R)}^\ast\leqrel{\leq}$ which is the definition of $\M*(R)$ being order-reflecting.
   \end{proof}

   \begin{lemma}
    $\pom$ is order-reflecting iff $\R(\pom)$ is exact.
   \end{lemma}
   \begin{proof}
    Here  we have:
    \begin{align*}
     {\models_{\R(\pom)}}\under{\models_{\R(\pom)}}
     &\eqrel \paren(y\then f^\ast)\under\paren(y\then f^\ast)\\
     &\eqrel f_\ast\then\paren(y\under y)\then f^\ast\\
     &\eqrel f_\ast\then y\then f^\ast.
    \end{align*}
    So again, the definition of order-preserving coincides with exactness.
   \end{proof}
\end{document}

%% file: macros.tex
\DeclareSymbolFont{stmry}{U}{stmry}{m}{n}
\DeclareMathDelimiter\llbracket{\mathopen}{stmry}{"4A}{stmry}{"71}
\DeclareMathDelimiter\rrbracket{\mathclose}{stmry}{"4B}{stmry}{"79}
\DeclareMathSymbol\leftrightarrowtriangle\mathbin{stmry}{"5D}
\DeclareMathSymbol\leftarrowtriangle\mathrel{stmry}{"5E}
\DeclareMathSymbol\rightarrowtriangle\mathrel{stmry}{"5F}
\DeclareMathSymbol\inplus\mathrel{stmry}{"41}
\DeclareMathSymbol\niplus\mathrel{stmry}{"42}
\DeclareMathSymbol\boxempty\mathbin{stmry}{"1F}

\DeclareFontEncoding{LS1}{}{}
\DeclareFontSubstitution{LS1}{stix}{m}{n}
\DeclareSymbolFont{sletters}       {LS1}{stix}     {m}{it}
\DeclareSymbolFont{symbols4}      {LS1}{stixbb}   {m}{it}
\DeclareMathSymbol{\Vdash}                    {\mathrel}{sletters}{"F6}
\DeclareMathSymbol{\vDash}                    {\mathrel}{sletters}{"F5}
\DeclareMathSymbol{\VDash}                    {\mathrel}{sletters}{"F8}
\DeclareMathSymbol{\vDdash}                   {\mathrel}{symbols4}{"B0}
\DeclareMathSymbol{\Dashv}                    {\mathrel}{symbols4}{"B2}
\DeclareMathSymbol{\DashV}                    {\mathrel}
{symbols4}{"B3}
\DeclareMathSymbol{\dashV}                    {\mathrel}{symbols4}{"B1}

\DeclareFontEncoding{FML}{}{}
\DeclareFontSubstitution{FML}{futm}{m}{it}\DeclareSymbolFont{letters}{FML}{futm}{m}{it}%
\DeclareMathSymbol{\varpartialdiff}{\mathord}{letters}{130}

\NewDocumentCommand\paren{r()}{\left(#1\right)}
\NewDocumentCommand\Paren{r()}{\left(\hspace{-2pt}\middle( #1\middle)\hspace{-2pt}\right)}
\RenewDocumentCommand\brack{r[]}{\left[#1\right]}
\newcommand\set[1]{\left\{#1\right\}}
\newcommand\setcompr[2]{\left\{#1~\middle|~#2\right\}}
\NewDocumentCommand\tuple{r<>}{\left\langle#1\right\rangle}
\NewDocumentCommand\Tuple{r<>}{\left\langle\hspace{-2pt}\middle\langle #1\middle\rangle\hspace{-2pt}\right\rangle}
\NewDocumentCommand\sem{r[]}{\left\llbracket {#1}\right\rrbracket}
\NewDocumentCommand\variables{r()}{\left\lfloor #1 \right\rfloor}
\NewDocumentCommand\size{r||}{\left|#1\right|}

\NewDocumentCommand\Functor{ m s d()}{%
    \ensuremath{\IfValueTF{#3}{\IfBooleanTF{#2}{{#1}{{#3}}}{{#1}{\paren(#3)}}}{#1}}\xspace%
}

\newcommand\f[1]{\Functor{\mathcal{#1}}}
\NewDocumentCommand\relarrow{ O{.08} r() O{midway, above} o r() }{
  \IfNoValueTF{#4}{
    \draw[->](#2) to (#5);
  }{
    \draw[->](#2) to node[#3] {#4} (#5);
  }
  \draw ($(#2)!.5!(#5)$) circle (#1)
}
\newcommand\relation[1][.5]{
  \mathrel{
    \tikz{
      \coordinate(s) at (0,0);
      \coordinate(t) at (#1,0);
      \relarrow(s)(t);
    }
  }
}

\NewDocumentCommand\implarr{r() r()}{
  \mathrel{
    \tikz{
      \draw[]($(#1.east) + (0,.03)$) to ($(#2.west) - (.03,-.03)$);
      \draw[]($(#1.east) - (0,.03)$) to ($(#2.west) - (.03,.03)$);
      \draw[]($(#2.west) + (-.06,.06)$) -- (#2.west) -- ($(#2.west) - (.06,.06)$);
    }
  }
}

\newcommand\converse{^{\circ}}
\newcommand\under{\mathbin{\backslash}}
\renewcommand\over{\mathbin{\slash}}
\newcommand\then{\mathbin{;}}

\newcommand\eqrel{~\leftrightarrowtriangle~}
\newcommand\leqrel{~\rightarrowtriangle~}
\newcommand\geqrel{~\leftarrowtriangle~}

\renewcommand{\models}{\vDash}

\newcommand\Repr{\ensuremath{\mathbf{Repr}}\xspace}
\newcommand\PoM{\ensuremath{\mathbf{ProM}}\xspace}
\newcommand\pset{\f{P}}

\newcommand\eqdef{\mathrel{{{\mathop:}{=}}}}

\NewDocumentCommand\univ{ s r() }{%
\IfBooleanTF{#1}{{#2}^{\forall}}{{\paren(#2)}^{\forall}}%
}
\renewcommand{\epsilon}{\varepsilon}
\renewcommand{\phi}{\varphi}

\newcommand\R{\f R}
\newcommand\M{\f M}

\NewDocumentCommand\var{r()}{{\mathrm{var}}\paren(#1)}
\newcommand{\pom}{\mathfrak{p}}